\documentclass[letterpaper,superscriptaddress,aps,prl,twocolumn,footinbib,floatfix]{revtex4-1}

\usepackage[verbose=true]{microtype}
\usepackage{amsmath}
\usepackage{amssymb}
\usepackage{times}
\usepackage{paralist}
\usepackage{quantum}
\usepackage[final]{graphicx}
\usepackage{color}
\usepackage[resetlabels]{multibib}
\usepackage[linktocpage,hypertexnames=false,pdfpagelabels]{hyperref}
\usepackage[amsmath,thmmarks,hyperref]{ntheorem}
\usepackage[capitalise]{cleveref}

\newtheorem{theorem}{Theorem}
\newtheorem{lemma}[theorem]{Lemma}

\theoremstyle{break}

\theoremstyle{break}
\theorembodyfont{\normalfont}

\theoremstyle{nonumberplain}
\theorembodyfont{\normalfont}
\theoremsymbol{$\Box$}
\newtheorem{proof}{Proof}

\newcites{supporting}{References}

\crefname{equations}{Eqs.}{Equations}
\creflabelformat{equations}{\textup{(#2#1#3)}}
\Crefname{problem}{Problem}{Problems}
\Crefname{condition}{condition}{conditions}
\Crefname{inparcondition}{condition}{conditions}
\Crefname{algorithm}{Algorithm}{Algorithm}

  {\par\noindent\ignorespacesafterend}
  {\par\noindent\ignorespacesafterend}
  {\par\noindent\ignorespacesafterend}

\DeclareMathOperator{\offdiag}{offdg}
\DeclareMathAlphabet{\mathitbf}{OML}{cmm}{b}{it}
\renewcommand{\vec}[1]{\mathitbf{#1}}

\newcommand{\keyword}[1]{\textit{#1}}
\newcommand{\vv}{\vec{v}}

\newcommand{\csproblem}[1]{\textsc{#1}}

\begin{document}

\title{Extracting dynamical equations from experimental data is NP-hard}
\author{Toby S.\ Cubitt}
\affiliation{Departamento de An\'alisis Matem\'atico, Universidad Complutense
  de Madrid, Plaza de Ciencias 3, Ciudad Universitaria, 28040 Madrid, Spain}
\author{Jens Eisert}
 \affiliation{Dahlem Center for Complex Quantum Systems, Freie Universit{\"a}t Berlin, 14195 Berlin, Germany}
\author{Michael M.\ Wolf}
\affiliation{Zentrum Mathematik, Technische Universit{\"at} M{\"u}nchen,
  85748 Garching, Germany}

\begin{abstract}
  The behavior of any physical system is governed by its underlying dynamical
  equations. Much of physics is concerned with discovering these dynamical
  equations and understanding their consequences. In this work, we show that,
  remarkably, identifying the underlying dynamical equation from any amount of
  experimental data, however precise, is a provably computationally hard
  problem (it is NP-hard), both for classical and quantum mechanical systems.
  As a by-product of this work, we give complexity-theoretic answers to both
  the quantum and classical embedding problems, two long-standing open
  problems in mathematics (the classical problem, in particular, dating back
  over 70 years).
\end{abstract}

\maketitle

A large part of physics is concerned with identifying the dynamical equations
of physical systems and understanding their consequences. But how do we deduce
the dynamical equations from experimental observations? Whether deducing the
laws of celestial mechanics from observations of the planets, determining
economic laws from observing monetary parameters, or deducing quantum
mechanical equations from observations of atoms, this task is clearly a
fundamental part of physics and, indeed, science in general. The task of
identifying dynamical equations from experimental data also turns out to be
closely related, in both the classical and quantum mechanical cases, to
long-standing open problems in mathematics (in the classical case, dating back
to 1937 \cite{Elfving}).

In this letter, we give complexity-theoretic solutions to both these open
problems. And these results lead to a surprising conclusion: regardless of how
much information one obtains through measuring a system, extracting the
underlying dynamical equations from those measurement data is in general an
intractable problem. More precisely, it is \keyword{NP-hard}. This means that
any computationally efficient method of determining which dynamical equations
are consistent with a set of measurement data would solve the (in)famous P
versus NP problem \cite{P=NP}, by implying that P$=$NP. Thus, if P$\neq$NP, as
is widely believed, there \emph{cannot} exist an efficient method of deducing
dynamical equations from any amount of experimental data. We also prove the
other direction: by reducing to an NP-complete problem we show that, if
P$=$NP, then there does exist an efficient algorithm for extracting dynamical
equations from experimental data. Thus the question of whether there exists an
efficient method for determining dynamical equations from measurement data is
equivalent to the P versus NP question.

Note that we are not restricting ourselves here to fundamental theories, where
other theoretical considerations may impose simplifications on the desired
form of the equations. We are also considering effective dynamical equations,
as encountered in the majority of experiments, where the full range of
possible dynamical equations can in principle be observed.

In the classical setting, the problem of extracting dynamical models from
experimental data has spawned an entire field known as \keyword{system
  identification} \cite{SystemIdentification}, which forms part of control
engineering -- after all, the precise knowledge of the dissipation is crucial
for actually understanding what control steps to apply. In the quantum case,
interest in understanding quantum dynamics, especially externally-induced
noise and decoherence, has been spurred on by efforts to develop quantum
information processing technology
\cite{implementations_review,Nielsen+Chuang}. Indeed, the primary goal of many
experiments is precisely to characterize and understand the dynamics of a
specific quantum system \cite{TeleportNature,Cory,Blatt,OBrien,Howard}. This
is precisely the task that we show to be computationally intractable in
general (assuming P$\neq$NP), both in quantum mechanics and in classical
physics.

\begin{figure}[htbp]
  \includegraphics[width=6cm]{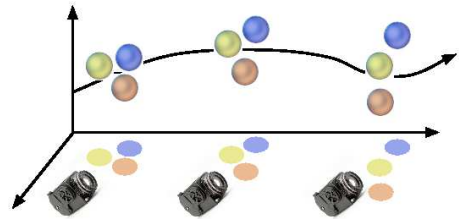}
  \caption{In an experiment, we can gather snapshots of the state of a
    physical system at various points in time. To understand the physics
    behind the system's behavior, we must reconstruct the underlying dynamical
    equations from the snapshots.}
  \label{fig:snapshots}
\end{figure}

\paragraph{Results.}
Let us make the task more concrete. We will throughout consider \emph{open
  system dynamics} which takes external influences and noise into account.
Recall that in classical mechanics, the most general state of a system is
described by a probability distribution $\vec{p}$ over its state space, which
for simplicity we will take to be finite-dimensional. Its evolution is then
described by a \keyword{master equation}, whose form is determined by the
system's Liouvillian, corresponding to a matrix $L$, as
  $\dot{\vec{p}} = L\vec{p}$.
The Liouvillian expresses interactions, conservation laws, external noise
etc., in short, it describes the underlying physics. In order for the
probabilities to remain positive and sum to one, the elements $L_{i,j}$
must obey two simple conditions \cite{Norris}: 
\begin{inparaenum}[(i)]
  \item $L_{i\neq j} \geq 0$, \label[inparcondition]{cond:Q_positivity}
  \item $\sum_i L_{i,j} = 0$. \label[inparcondition]{cond:Q_normalisation}
\end{inparaenum}

In the quantum setting, the density matrix $\rho$ plays the analogous role to
that of the classical probability distribution, but the quantum master
equations are still determined by a Liouvillian:
\begin{equation}\label{eq:quantum_master_equation}
  \dot{\rho} = \mathcal{L}(\rho).
\end{equation}
In his seminal 1976 paper \cite{Lindblad}, Lindblad established the general
form that any quantum Liouvillian must take if it is to generate a
completely-positive trace-preserving evolution (so that density matrices
always evolve into density matrices, directly analogous to probabilities
remaining positive and normalised in the classical case):
\begin{equation}\label{eq:Lindblad_form}
  \mathcal{L}(\rho)
  = i[\rho,H] + \sum_{\alpha,\beta} G_{\alpha,\beta}\Bigl(
      F_\alpha\rho F_\beta^\dagger
      - \frac{1}{2}\{F_\beta^\dagger F_\alpha,\rho\}_+
    \Bigr).
\end{equation}
Here, $H$ is the Hamiltonian of the system, $G$ is a positive semi-definite
matrix and, along with the matrices $F_\alpha$, describes decoherence
processes. ($[.,.]$ and $\{.,.\}_+$ denote
respectively the commutator and anti-commutator.) These master equations of
\keyword{Lindblad form} have become the mainstay of the dynamical theory of
open quantum systems, and are crucial to the description of quantum mechanics
experiments \cite{Carmichael}. In principle, the Liouvillian could itself be
time-dependent, describing a system whose underlying physics is changing over
time. Here, we restrict our attention to the problem of
finding a time-independent Liouvillian, as this is a good assumption for
experiments in which external parameters are held constant. The more
general time-dependent problem is expected to be harder still.

What is the best possible data that an experimentalist can conceivably gather
about an evolving system? At least in principle, they can repeatedly prepare
the system in any chosen initial state, allow it to evolve for some period of
time, and then perform any measurement. In fact, for a careful choice of
initial states and measurements, it is possible in this way to reconstruct a
complete ``snapshot'' of the dynamics at any particular time. In the quantum
setting, this technique is known as \keyword{quantum process tomography}
\cite{Nielsen+Chuang}. Quantum process tomography is now routinely carried out
in many different physical systems, from NMR \cite{TeleportNature,Cory} to
trapped ions \cite{Blatt}, from photons \cite{OBrien} to solid-state devices
\cite{Howard}.

A tomographic snapshot tells us \emph{everything} there is to know about the
evolution at the time $t$ when the snapshot was taken. Each snapshot is a
dynamical map $\mathcal{E}_t$, which describes how the initial state,
$\vec{p}_0$ or $\rho_0$, is transformed into $\vec{p}(t) =
\mathcal{E}_t(\vec{p}_0)$ or $\rho(t) = \mathcal{E}_t(\rho_0)$. Any
measurement at time $t$ can be viewed as an imperfect version of process
tomography, giving partial information about the snapshot, and the outcome of
any measurement of the system at time $t$ can be predicted once
$\mathcal{E}_t$ is known. Thus the most complete data that can be gathered
about a system's dynamics consists of a set of snapshots taken at a sequence
of different points in time.

Let us concentrate first on the quantum case. Quantum dynamical maps
$\mathcal{E}_t$ are described mathematically by completely positive,
trace-preserving (CPT) maps \cite{Nielsen+Chuang} (also known as
\keyword{quantum channels}). The problem of deducing the dynamical equations
from measurement data is then one of finding a Lindblad master
equation~\eqref{eq:quantum_master_equation} that accounts for the CPT
snapshots $\mathcal{E}_t$. This is essentially the converse problem to that
considered by Lindblad \cite{Lindblad,Gorini}. Given its relevance, it is not
surprising that numerous heuristic numerical techniques have been applied to
tackle this problem \cite{Cory,Howard2}. But unfortunately these give no
guarantee as to whether a correct answer has been found. Our results show that
the failure of these heuristic techniques is an inevitable consequence of the
inherent intractability of the problem.

Before tackling the problem of finding dynamical equations, let us start by
considering an apparently much simpler question: given a \emph{single}
snapshot $\mathcal{E}$, does there even \emph{exist} a Liouvillian
$\mathcal{L}$ that could have generated it? Not every CPT map $\mathcal{E}$
can be generated by a Lindblad master equation
\cite{divisibility,markovianity}, so the question of the existence of such a
Liouvillian (\cref{eq:Lindblad_form}) is a well-posed problem. A dynamical map
that \emph{is} generated by a Lindblad form Liouvillian is said to be
\keyword{Markovian}, so this problem is sometimes referred to as the
\keyword{Markovianity problem}. Non-Markovian snapshots \cite{Structured} can
arise if the environment carries a memory of the past, so that the system's
evolution cannot be described by \cref{eq:quantum_master_equation} in the
first place, as that assumes the system is sufficiently isolated from its
environment for its dynamics to be described independently.

It is important to note that, for the results to apply to real experimental
data, we must take into account the fact that a snapshot can only ever be
measured up to some experimental error. We should therefore be satisfied if we
can answer the question for some approximation $\mathcal{E'}$ to the measured
snapshot $\mathcal{E}$, as long as the approximation is accurate up to
experimental error. Mathematically, this is known as a \keyword{weak
  membership} formulation of the problem.

To address the Markovianity problem, we will require some basic concepts from
complexity theory. Recall that P is the class of computational problems that
can be solved efficiently on a classical computer. The class NP instead only
requires an efficient verification of solutions, and contains problems that
are believed to be impossible to solve efficiently, such as the famous
\csproblem{3SAT} problem, and the travelling salesman problem. A problem is
\keyword{NP-hard} if solving it efficiently would also lead to efficient
solutions to \emph{all} other NP problems. A problem that is both NP-hard and
is also itself in the class NP is said to be \keyword{NP-complete}. The
\csproblem{3SAT} and travelling salesman problems are both examples of
NP-complete problems, whereas the problem of factoring large integers is an
example of an NP problem that is believed not to be NP-hard
\cite{Computers+Intractability}.

Rather than considering \csproblem{3SAT}, it is more convenient here to
consider the equivalent \csproblem{1-in-3SAT} problem, into which
\csproblem{3SAT} can easily be transformed \cite{Computers+Intractability},
and which is therefore also NP-complete. We will show that any instance of the
\csproblem{1-in-3SAT} problem can be efficiently transformed into an instance
of the Markovianity problem (see also \cite{supporting}), thus proving that
the latter is at least as hard as \csproblem{1-in-3SAT}; any efficient
procedure for determining whether a snapshot has some underlying Liouvillian
would immediately imply an efficient procedure for solving
\csproblem{1-in-3SAT}. But \csproblem{1-in-3SAT} is NP-complete, so this would
immediately give an efficient algorithm for solving any NP-problem, implying
P$=$NP.
However, as discussed above, the Markovianity problem is just a special case
of the more general---and more important---problem of extracting the
underlying dynamical equations from experimental data. If P$\neq$NP, as is
widely believed, then there \emph{cannot exist a computationally efficient
  method of deducing dynamical equations from any amount of experimental
  data.}

We can go further than this. Through the relation to NP-complete problems such
as \csproblem{1-in-3SAT}, we can reduce the Markovianity problem to the task
of solving an NP-complete problem. This gives the first rigorous, provably
correct algorithm for extracting the underlying dynamical equations from a set
of experimental data, albeit one that is necessarily inefficient for systems
with more than a few degrees of freedom (otherwise we would have proven
P$=$NP!).

We have focussed so far on the more complex case of quantum systems, and one
might perhaps expect that systems governed by classical physics would be
easier to analyse. However, essentially the same argument proves that exactly
the same results hold for classical systems, too. (See also
\cite{supporting}.)

\paragraph{The technical argument.}
It is convenient to represent a snapshot $\mathcal{E}$ of the dynamics of a
quantum system (a CPT map) by a matrix $E$,
\begin{equation}
  E_{i,j;k,l} =
  \tr\bigl[\mathcal{E}\bigl(\ketbra{i}{j}\bigr)\cdot\ketbra{k}{l}\bigr]
\end{equation}
(the row- and column-indices of $E$ are the double-indices $i,j$ and $k,l$,
respectively). Looked at this way, each measurement that is performed pins
down the values of some of these matrix elements \cite{Nielsen+Chuang}. A
snapshot of a Markovian evolution is then one with a Liouvillian $\mathcal{L}$
(represented in the same way by a matrix $L$) such that $E = e^L$,
and, for all times $t\geq 0$, $E_t=e^{Lt}$ are also valid quantum
dynamical (CPT) maps.

The Markovianity problem can be transformed into an equivalent question about
the Liouvillian. Inverting the relationship $E=e^L$, we have $L = \log E$.
There are, however, infinitely many possible branches of the logarithm, since
the phases of complex eigenvalues of $E$ are only defined modulo $2\pi i$. The
problem then becomes one of determining whether \emph{any one of these} is a
valid Liouvillian (i.e.\ of Lindblad form~\eqref{eq:Lindblad_form}).
This translates into the following necessary and sufficient conditions on the
matrix $L$ \cite{markovianity}:
\begin{enumerate}
\item $L^\Gamma$ is Hermitian, where $\Gamma$ is defined by its action on
  basis elements: $\ketbra{i,j}{k,l}^\Gamma =
  \ketbra{i,k}{j,l}$. \label[inparcondition]{cond:L_Hermiticity}
\item $L$ fulfils the normalisation $\bra{\omega}L=0$, where
  $\ket{\omega}=\sum_i\ket{i,i}/{\sqrt{d}}$ is maximally entangled.
  \label[inparcondition]{cond:L_normalisation}
\item $L$ satisfies \keyword{conditional complete positivity} (ccp), i.e.
  \label[inparcondition]{cond:L_ccp}
  $(\id-\omega)L^\Gamma(\id-\omega)\geq 0,\,\,\omega=\proj{\omega}$.
\end{enumerate}
All branches $L_m$ of the logarithm can be obtained by adding integer
multiples of $2\pi i$ to the eigenvalues of the principle branch $L_0$, so we
can parametrise all the possible branches by a set of integers $m_c$:
\begin{align}
     L_m &= \log E =
      L_0 + \sum_c m_c A^{(c)},\label{eq:logs}\\
            A^{(c)} &= 2\pi i \bigl(\ketbra{l_c}{r_c}
             - \mathbbm{F}(\ketbra{l_c}{r_c})\bigr),\label{eq:A_c}
\end{align}
with $\ket{l_c}$ and $\bra{r_c}$ the left- and right-eigenvectors of $E$.
$\mathbbm{F}$ is the operation
$\mathbbm{F}(\ketbra{i,j}{k,l})=\ketbra{j,i}{l,k}^{\!*}$, where $^*$ denotes
the complex-conjugate, and we have already restricted the parametrisation to
logarithms that satisfy \cref{cond:L_Hermiticity}.

We will prove that this \keyword{Liouvillian problem} is NP-hard, by showing
how to encode any instance of the NP-complete \csproblem{1-in-3SAT} problem
into it. Recall that the task in \csproblem{1-in-3SAT} is to determine whether
a given logical expression can be satisfied or not. The expression is made up
of ``clauses'', all of which must be satisfied simultaneously. Each clause
involves three boolean variables (variables with values ``true'' or
``false''), which can be represented by integers $m_c=0,1$.
%
In \csproblem{1-in-3SAT}, a clause is satisfied if and only if \emph{exactly
  one} of the variables appearing in the clause is true (as opposed to
\csproblem{3SAT}, in which \emph{at least one} must be true), and no boolean
negation is necessary. Note that, in terms of integer variables $m_c$, a
\csproblem{1-in-3SAT} clause containing variables $m_i$, $m_j$ and $m_k$ can
be expressed as
\begin{subequations}\label[equations]{eq:1-in-3SAT_constraints}
\begin{gather}
  1 \leq m_i + m_j + m_k \leq 1, \label{eq:clause_constraints}\\
  0 \leq m_i,m_j,m_k \leq 1. \label{eq:boolean_constraints}
\end{gather}
\end{subequations}

\emph{If the matrices appearing in
  \cref{cond:L_Hermiticity,cond:L_normalisation,cond:L_ccp} were diagonal},
\cref{cond:L_ccp} would give us a concise way of writing the coefficients and
constants of a set of inequalities such as \cref{eq:1-in-3SAT_constraints} in
the diagonal elements. However, the problem we are facing here is
significantly more challenging: diagonal matrices will never satisfy
\cref{cond:L_Hermiticity,cond:L_normalisation}, and the matrices $L_0$ and
$A^{(c)}$ cannot be chosen independently, since they are determined by the
eigenvectors and eigenvalues of a single matrix $E$.

These substantial obstacles can be overcome, however. The key step in encoding
the above boolean constraints in a quantum Liouvillian is to restrict our
attention to matrices $L_0$ and $A^{(c)}$ with the following special forms:
\begin{gather}
  L_0 = 2\pi\sum_{i,j} Q_{i,j}\ketbra{i,i}{j,j}
         + 2\pi\sum_{i\neq j}P_{i,j}\ketbra{i,j}{i,j},
         \label{eq:special_case_L0}\\
  A^{(c)} = 2\pi\sum_{i\neq j} B^{(c)}_{i,j}\ketbra{i,i}{j,j},
         \label{eq:special_case_Ac}
\end{gather}
with coefficient matrices
\begin{align}
  Q &= \sum_r\vv_r^{\vphantom{T}}\vv_r^T
         \otimes\begin{pmatrix}1&1\\1&1\end{pmatrix}
         \otimes\begin{pmatrix}
           k+\lambda_r&\lambda_r\\\lambda_r&k+\lambda_r
         \end{pmatrix}
         \notag\\
       &\qquad
       +\sum_{c}\vv_c^{\vphantom{T}}\vv_c^T\otimes
         \begin{pmatrix}\phantom{-}1&-1\\-1&\phantom{-}1\end{pmatrix}
         \otimes\begin{pmatrix}
           k&-\frac{1}{3}\\\frac{1}{3}&\phantom{-}k
         \end{pmatrix}
         \label{eq:special_case_Q}\\
       &\qquad
       +\sum_{c'}\vv_{c'}^{\vphantom{T}}\vv_{c'}^T\otimes
         \begin{pmatrix}\phantom{-}1&-1\\-1&\phantom{-}1\end{pmatrix}
         \otimes\begin{pmatrix}k&0\\0&k\end{pmatrix}\notag,\\
  B^{(c)} &= \vv_c^{\vphantom{T}}\vv_c^T\otimes
         \begin{pmatrix}\phantom{-}1&-1\\-1&\phantom{-}1\end{pmatrix}
         \otimes
         \begin{pmatrix}\phantom{-}0&1\\-1&0\end{pmatrix}.
         \label{eq:special_case_Bc}
\end{align}
The sets of real vectors $\{\vv_r\}$ and $\{\vv_c,\vv_{c'}\}$ should each form
an orthogonal basis, and the parameters $k$, $\lambda_r$ and $P_{i,j}$ are
also real. The advantage of this restriction is that the action of the
$\Gamma$ operation on matrices of this form is somewhat easier to analyse, as
can readily be seen from its definition (given in \cref{cond:L_Hermiticity},
above).

It is a simple matter to verify that the eigenvalues and eigenvectors of $L_0$
and $B^{(c)}$ do indeed parametrise the logarithms of a matrix $E$, and that
the Hermiticity and normalisation conditions
\cref{cond:L_Hermiticity,cond:L_normalisation} necessary for $L$ to be a
valid quantum Liouvillian are indeed satisfied by the forms given in
\cref{eq:special_case_L0,eq:special_case_Ac,eq:special_case_Q,eq:special_case_Bc},
as long as $\vec{w}^TQ = 0$ and $\diag(P)^\Gamma$ is Hermitian (where for
$d$--dimensional $Q$, $\vec{w}=(1,1,\dots,1)^T/\sqrt{d}$, and $\diag(P)$
denotes the $d^2$\nobreakdash--dimensional matrix with $P_{i,j}$ down its main
diagonal). Furthermore, the ccp condition \cref{cond:L_ccp} reduces for
this special form to the pair of conditions:
\begin{subequations}
\begin{gather}
  \sum_c B^{(c)}_{i,j}\,m_c + Q_{i,j} \geq 0 \qquad i\neq j,
  \label{eq:special_case_ccp1}\\
  \left(\id-\vec{w}\vec{w}^T\right)\left(\diag Q + \offdiag P\right)
    \left(\id-\vec{w}\vec{w}^T\right) \geq 0,
  \label{eq:special_case_ccp2}
\end{gather}
\end{subequations}
where $M = (\diag Q + \offdiag P)$ denotes the $d$\nobreakdash--dimensional
matrix with diagonal elements $M_{i,i}=Q_{i,i}$ and off-diagonal elements
$M_{i\neq j}=P_{i,j}$.

We now encode the coefficients of the 1-in-3SAT problem from
\cref{eq:1-in-3SAT_constraints} into the elements of $\vv_c$. For each clause
in \cref{eq:clause_constraints}, write a ``1'' in a new element of $\vv_i$,
$\vv_j$ and $\vv_k$, and a ``0'' in the corresponding element of all other
$\vv_c$'s. For each $\vv_c$, write a ``1'' in a new element of the vector,
writing a ``0'' in the corresponding element of all the other $\vv_c$'s (these
elements will be used to restrict each $m_c$ to the values 0 or 1). Finally,
extend the vectors so that they are mutually orthogonal and have the same
length, which can always be done. One can now verify directly that, by
choosing appropriate $\vv_r$, \cref{eq:1-in-3SAT_constraints} are equivalent
to the 1-in-3SAT inequalities of \cref{eq:special_case_ccp2}. Furthermore,
\cref{cond:L_Hermiticity,cond:L_normalisation} are always satisfied. (See
\cite{supporting} for more detail.) Thus we have succeeded in encoding
1-in-3SAT into the Liouvillian problem. As the latter is equivalent to the
Markovianity problem, this proves that the Markovianity problem is itself
NP-hard.
This construction easily generalizes to the original question of
\emph{finding} which dynamical equations (if any) could have generated a given
set of snapshots \cite{supporting}: any method of finding dynamical equations
consistent with the data would obviously also answer the question of whether
these exist, allowing us to solve all NP problems.

Note that, on the positive side, by carrying out a brute-force search for
solutions of the corresponding satisfiability problem (in the case considered
above, this is \csproblem{1-in-3SAT}, but more generally it is an integer
semi-definite constraint problem defined by
\cref{cond:L_ccp,cond:L_normalisation,cond:L_Hermiticity}, which is obviously
in NP), we immediately obtain an algorithm for extracting dynamical equations
from measurement data that is guaranteed to give the correct answer. Although
such an algorithm will not work in practice even for moderately complex
systems, the NP-hardness proves that we cannot hope for an efficient algorithm
(unless P$=$NP). And it \emph{can} be applied to systems with few degrees of
freedom, making it immediately applicable at least to many current quantum
experiments.

What of the classical setting? The classical analogue of the Markovianity
problem is the so-called \keyword{embedding problem} for stochastic matrices,
originally posed in 1937 \cite{Elfving}. Despite considerable effort
\cite{Kingman} the general problem has, however, remained open until now
\cite{Mukherjea}. Strictly speaking, the quantum result does not directly
imply anything about the classical problem. Nevertheless, the arguments we
have given in the more complicated quantum setting can straightforwardly be
adapted to the classical embedding problem \cite{supporting}, proving that
this is NP-hard, too. (See \cite{supporting} for details.)

\paragraph{Discussion.}
On the one hand, this work leads to a rigorous algorithm for extracting the
underlying dynamical equations from experimental data. For systems with few
effective degrees of freedom, as encountered for example in all quantum
tomography experiments to date \cite{TeleportNature,Cory,Blatt,OBrien,Howard},
this gives the first practical and provably correct algorithm for this key
task. For systems with many degrees of freedom, the algorithm is necessarily
inefficient, with a run-time that scales exponentially. But our
complexity-theoretic NP-hardness results show that we cannot hope for a
polynomial-time algorithm. Note also that the hardness cannot be attributed to
allowing high-energy processes in the dynamics (high branches of the
logarithm), as the reduction from the \csproblem{1-in-3SAT} problem only needs
low-energy dynamics ($m$ is restricted to $0$ or $1$).

On the other hand, our results also prove that for general systems, deducing
the underlying dynamical equations from experimental data is computationally
intractable, unless one can show that P$=$NP. This hardness result is true
whether the system is quantum or classical, and regardless of how much
experimental data we gather about the system.
These results also imply that various closely related problems, such as
finding the dynamical equation that best approximates the data, or testing a
dynamical model against experimental data, are also intractable in general, as
any method of solving these problems could easily be used to solve the
original problem.

Experience would seem to suggest that, whilst general classical and quantum
dynamical equations may be impossible to deduce from experimental data, the
dynamics that we actually encounter are typically much easier to analyse. Our
results pose the interesting question of why this should be, and whether there
is some general physical principle that rules out intractable dynamics.

\paragraph{Acknowledgements.} The authors would like to thank J.\ I.\ Cirac,
A.\ Winter, C.\ Goldschmidt, and J.\ Martin for valuable discussions. This
work has been supported by a Leverhulme early career fellowship, by the EU
(QAP, QESSENCE, MINOS, COMPAS, COQUIT, QUEVADIS), by Spanish grants QUITEMAD,
I-MATH, and MTM2008-01366, by the EURYI, the BMBF (QuOReP), and the Danish
Research Council (FNU).

\bibliography{Markov_complexity}

\onecolumngrid  
\clearpage
\twocolumngrid

\appendix
\setcounter{page}{1}
\setcounter{equation}{0}

\section*{\large Supporting Material}

\subsection*{Encoding \csproblem{3SAT} in a Liouvillian}
\noindent We start from the special form for the matrices $L_0$ and $A^{(c)}$
defined in the main text:
\begin{gather}
  L_0 = 2\pi\sum_{i,j} Q_{i,j}\ketbra{i,i}{j,j}
         + 2\pi\sum_{i\neq j}P_{i,j}\ketbra{i,j}{i,j},
         \label{eq:epaps:special_case_L0}\\
  A^{(c)} = 2\pi\sum_{i\neq j} B^{(c)}_{i,j}\ketbra{i,i}{j,j},
         \label{eq:epaps:special_case_Ac}
\end{gather}
with
\begin{align}
  Q &= \sum_r\vv_r^{\vphantom{T}}\vv_r^T
         \otimes\begin{pmatrix}1&1\\1&1\end{pmatrix}
         \otimes\begin{pmatrix}
           k+\lambda_r&\lambda_r\\\lambda_r&k+\lambda_r
         \end{pmatrix}
         \notag\\
       &\qquad
       +\sum_{c}\vv_c^{\vphantom{T}}\vv_c^T\otimes
         \begin{pmatrix}\phantom{-}1&-1\\-1&\phantom{-}1\end{pmatrix}
         \otimes\begin{pmatrix}
           k&-\frac{1}{3}\\\frac{1}{3}&\phantom{-}k
         \end{pmatrix}
         \label{eq:epaps:special_case_Q}\\
       &\qquad
       +\sum_{c'}\vv_{c'}^{\vphantom{T}}\vv_{c'}^T\otimes
         \begin{pmatrix}\phantom{-}1&-1\\-1&\phantom{-}1\end{pmatrix}
         \otimes\begin{pmatrix}k&0\\0&k\end{pmatrix}\notag,\\
  B^{(c)} &= \vv_c^{\vphantom{T}}\vv_c^T\otimes
         \begin{pmatrix}\phantom{-}1&-1\\-1&\phantom{-}1\end{pmatrix}
         \otimes
         \begin{pmatrix}\phantom{-}0&1\\-1&0\end{pmatrix}.
         \label{eq:epaps:special_case_Bc}
\end{align}
Recall that the ccp condition~(iii), 
given on page~3 of the main text, 
reduces for this special form to the pair of conditions:
\begin{subequations}
\begin{gather}
  \sum_c B^{(c)}_{i,j}\,m_c + Q_{i,j} \geq 0 \qquad i\neq j,
  \label{eq:epaps:special_case_ccp1}\\
  \left(\id-\vec{w}\vec{w}^T\right)\left(\diag Q + \offdiag P\right)
    \left(\id-\vec{w}\vec{w}^T\right) \geq 0.
  \label{eq:epaps:special_case_ccp2}
\end{gather}
\end{subequations}

As explained in the main text, we encode a \csproblem{1-in-3SAT} problem into
these matrices by writing the clauses into the vectors $\vv_c$. Denote the
total number of variables and clauses by $V$ and $C$, respectively. For each
clause $n$ involving the $i^\mathrm{th}$, $j^\mathrm{th}$ and $k^\mathrm{th}$
boolean variables, write a ``1'' in the $n^\mathrm{th}$ element of $\vv_i$,
$\vv_j$ and $\vv_k$, and write a ``0'' in the same element of all the other
$\vv_c$'s. Now, for each $\vv_c$, write a ``1'' in its $C+c$'th element,
writing a ``0'' in the corresponding element of all the other vectors.
Finally, extend the vectors so that they are mutually orthogonal and have the
same length, which can always be done. This produces vectors with at most
$C+2V$ elements.

This procedure encodes the coefficients for the \csproblem{1-in-3SAT}
inequalities into some of the on-diagonal $4\times 4$ blocks of the $B^{(c)}$
matrices. Specifically, if we imagine colouring $B^{(c)}$ in a chess-board
pattern (starting with a ``white square'' in the top-leftmost element), then
the coefficients for one \csproblem{1-in-3SAT} constraint from
Eq.~(7) 
of the main text are duplicated in all the ``black squares'' of one diagonal
$4\times 4$ block.

Colouring $Q$ in the same chess-board pattern, the contribution to its ``black
squares'' from the first term of \cref{eq:epaps:special_case_Q} is generated
by the off-diagonal elements $\lambda_r$:
\begin{equation}
  \label{eq:epaps:symmetric_matrix}
  \sum_r\vec{v}_r^{\vphantom{T}}\vec{v}_r^T\otimes
  \begin{pmatrix}1&1\\1&1\end{pmatrix}\otimes
  \begin{pmatrix}\cdot&\lambda_r\\\lambda_r&\cdot\end{pmatrix}
  = S\otimes
    \begin{pmatrix}1&1\\1&1\end{pmatrix}\otimes
    \begin{pmatrix}\cdot&1\\1&\cdot\end{pmatrix}.
\end{equation}

Since $\vv_r$ and $\lambda_r$ can be chosen freely, the first tensor factor in
this expression is just the eigenvalue decomposition of an arbitrary real,
symmetric matrix $S$. If we choose the first $C$ diagonal elements of $S$ to
be $1/2$, and choose the next $V$ diagonal elements to be $5/6$, then it is
straightforward to verify that the equations in the ccp condition of
\cref{eq:epaps:special_case_ccp1} corresponding to the ``black squares'' in
on-diagonal $4\times 4$ blocks are given by
\begin{equation}
  \label[equations]{eq:epaps:inequalities}
  \begin{aligned}
    m_i,m_j,m_k \geq -\frac{1}{2},&\quad -m_i,m_j,m_k\geq -\frac{7}{6},\\
    m_i + m_j + m_k \geq \frac{1}{2}, &\quad
    -m_i - m_j - m_k \geq -\frac{3}{2},
  \end{aligned}
\end{equation}
for all $m_i$, $m_j$, $m_k$ appearing together in a \csproblem{1-in-3SAT}
clause. Since the $m_c$ are integers, these inequalities are exactly
equivalent to the \csproblem{1-in-3SAT} constraints given in Eq.~(7) of the
main text.

We have successfully encoded the correct coefficients and constants into
certain matrix elements of $B^{(c)}$ and $Q$. But all the other elements of
these matrices also generate inequalities via
\cref{eq:epaps:special_case_ccp1}. To ``filter out'' these unwanted
inequalities, we choose the remaining diagonal elements and all off-diagonal
elements of the symmetric matrix $S$ to be large and positive, thereby
ensuring all unwanted inequalities are always trivially satisfied.

$L_m$, as constructed so far, will not satisfy the normalisation
condition~(ii) 
given on page~3 of the main text. For that, we need to ensure that
$\vec{w}^TQ=0$, i.e., that the columns of $Q$ sum to zero. We use the ``white
squares'' of $Q$, generated by the diagonal elements in the third tensor
factors of \cref{eq:epaps:special_case_Q}, to renormalise these column sums to
zero. Recall that both $\{\vv_r\}$ and $\{\vv_c,\vv_{c'}\}$ are complete sets
of mutually orthogonal vectors. Rearranging \cref{eq:epaps:special_case_Q},
$Q$ is therefore given by
\begin{multline}
  \label{eq:epaps:constructed_Q}
  Q = k\id + S\otimes\begin{pmatrix}1&1\\1&1\end{pmatrix}
             \otimes\begin{pmatrix}1&1\\1&1\end{pmatrix}\\
      + \sum_c\vv_c^{\vphantom{T}}\vv_c^T
        \otimes\begin{pmatrix}
          \phantom{-}1&-1\\-1&\phantom{-}1
        \end{pmatrix}
        \otimes\begin{pmatrix}
          0&-\frac{1}{3}\\\frac{1}{3}&\phantom{-}0
        \end{pmatrix},
\end{multline}
where $\id$ is the identity matrix. Now, the only requirement on the
off-diagonal elements of $S$ is that they be sufficiently positive to filter
out the unwanted inequalities. Also, from the form of
\cref{eq:epaps:constructed_Q}, the columns in any individual $4\times 4$ block
of $Q$ sum to the same value. Thus, by adjusting the elements of $S$, we can
ensure that all columns of $Q-k\id$ sum to the \emph{same} positive value,
$\sigma$ say. Choosing $k=-\sigma$, the negative on-diagonal element in each
column generated by the $k\id=-\sigma\id$ term will cancel the positive
contribution from the other terms, thereby satisfying the normalisation
condition, as required.

Finally, we must ensure that the second ccp condition from
\cref{eq:epaps:special_case_ccp2} is always satisfied, for which we require
the following Lemma:
\begin{lemma}
  \label{lem:epaps:Q+P}
  If $Q=-k\id$ is $d$-dimensional, then for any real $k$ there exists a matrix
  $P$ such that $\diag P = 0$ and
  \begin{equation}
    (\id-\vec{w}\vec{w}^T)(Q+P)(\id-\vec{w}\vec{w}^T)\geq 0,
  \end{equation}
  where $\vec{w} = (1,1,\dots,1)^T/\sqrt{d}$.
\end{lemma}
\begin{proof}
  Choose $P = \alpha(\id-\vec{w}\vec{w}^T) +
  \alpha(1-d)\vec{w}\vec{w}^T$. Then the diagonal elements of $P$ are
  \begin{equation}
    P_{i,i}
    = \alpha\left(1-\frac{1}{d}\right) + \alpha(1-d)\frac{1}{d}
    = 0,
  \end{equation}
  and
  \begin{equation}
    (\id-\vec{w}\vec{w}^T)(Q+P)(\id-\vec{w}\vec{w}^T)
    = (\alpha-k)(\id-\vec{w}\vec{w}^T),
  \end{equation}
  which is positive semi-definite for $\alpha\geq k$.
\end{proof}
The coefficients $P_{i,j}$ in \cref{eq:epaps:special_case_L0} can be chosen
freely, since they play no role in either the normalisation or in encoding
\csproblem{1-in-3SAT}, so the $[\offdiag P]$ term in the ccp condition of
\cref{eq:epaps:special_case_ccp2} can be chosen to be any matrix with zeros
down its main diagonal. Also, from \cref{eq:epaps:constructed_Q}, all diagonal
elements of $Q$ are equal to $k = -\sigma$. Thus
\cref{eq:epaps:special_case_ccp2} is exactly of the form given in
\cref{lem:epaps:Q+P}, and choosing $P$ accordingly ensures that it is always
satisfied. Furthermore, since gives a $P^\Gamma$ that is Hermitian,
condition~(i) of the main text is automatically also satisfied.

We have constructed $L_0$ and $A^{(c)}$ such that there exists an $L_m$
satisfying conditions~(i), (ii)
and~(iii) 
from page~3 of the main text if (and only if) the original
\csproblem{1-in-3SAT} instance was satisfiable. But we have already shown that
condition~(iii), 
along with conditions~(i) and~(ii), 
are satisfied if (and only if) $L_m$ is of Lindblad form, which in turn is
equivalent to $E=e^{L_m}=e^{L_0}$ being Markovian.

Furthermore, the integer solutions of \cref{eq:epaps:inequalities} are
insensitive to small perturbations of the coefficients and constants, so any
sufficiently good approximation $E'$ will still be Markovian if $E$ is, and
vice versa, as long as we impose sufficient precision requirements. Indeed, it
is natural to expect that if a snapshot $E$ is close to being Markovian, it
will have a generator $L_m$ that is close to being of Lindblad form. Making
this rigorous is less trivial, but follows from continuity properties of the
matrix exponential \citesupporting{Horn+Johnson_Topics} and logarithm
\citesupporting{Weilenmann}. The Markovianity problem is therefore equivalent
to the problem of determining whether any $L_m$ obeys the three conditions~(i)
to~(iii), 
up to the necessary approximation accuracy. Thus we have successfully encoded
\csproblem{1-in-3SAT} into the Liouvillian problem, such that the
corresponding snapshot $E$ is Markovian if (and only if) the
\csproblem{1-in-3SAT} instance was satisfiable.

Using standard perturbation theory results for eigenvalues and eigenvectors
\citesupporting{Horn+Johnson,Golub+vanLoan}, a careful analysis reveals that a
precision of $\order{V^{-1}(C+2V)^{-3}}$ is sufficient, which scales only
polynomially with the number of degrees of freedom in the system (i.e., with
the size of the Liouvillian matrix). Though a polynomial scaling is not
strictly speaking necessary to prove NP-hardness, it makes the result more
compelling, as it shows that the complexity does not result from demanding
unreasonable precision requirements. This is sometimes called \keyword{strong
  NP-hardness} of a weak-membership problem (cf.\
Ref.~\citesupporting{Gharibian}).

This so-called \keyword{weak-membership} formulation of the problem---allowing
for approximate answers---is vital if the question is to be reasonable from an
experimental perspective: the snapshot $E$ can only be measured up to some
experimental error. Allowing for approximate answers can only make the problem
easier than requiring an exact answer, so the fact that the problem remains
NP-hard even for finite (even polynomial) precision is crucial to the
experimental relevance of the hardness result. In fact, the weak-membership
formulation is also necessary from a theoretical perspective. If $E$ happened
to be close to the boundary of the set of Markovian maps, then it would be
close to both Markovian and non-Markovian maps, and an exact answer could
require the matrix elements of $E$ to be specified to infinite precision,
which is not reasonable even theoretically.

\subsection*{Several snapshots}
Clearly, if we can \emph{find} a set of dynamical equations whenever they
exist, we can also determine \emph{whether} they exist. So finding the
dynamical equations is at least as hard as answering the existence question.
For a single snapshot, the latter is just the Markovianity problem again. But,
having constructed $L_0$ and $A^{(c)}$ as described above, it is easy to
generalise this to any number of snapshots $\mathcal{E}_t$: simply take $E_t =
e^{L_0 t}$ for as many different times $t$ as desired.

\subsection*{The classical setting}
The analogue of the Markovianity problem in the classical setting is known as
the \keyword{embedding problem}. Given a stochastic matrix, this asks whether
it can be generated by any continuous, time-homogeneous Markov process (i.e.,
by dynamics obeying a time-independent classical master equation). The quantum
mechanical proof described above does not directly imply anything about the
classical problem (nor vice versa). Nevertheless, it turns out that the
arguments used in the quantum setting can readily be adapted to the classical
embedding problem.

We can reduce the embedding problem to a question about the (classical)
Liouvillian, in the same way as in the quantum case. Comparing the conditions
for $L$ to be a valid classical Liouvillian (see conditions~(i) and~(ii) on
page~1 of the main text) with the matrices $Q$ and $B^{(c)}$ from
\cref{eq:epaps:special_case_Q,eq:epaps:special_case_Bc}, we see that $Q+\sum
m_c B^{(c)}$ is a valid \emph{classical} Liouvillian if and only if the
\csproblem{1-in-3SAT} problem was satisfiable. In other words, for the
classical case, we simply need to use the matrices $Q$ and $B^{(c)}$, rather
than the full matrices $L_0$ and $A^{(c)}$ used in the quantum construction.
The rest of the arguments proceed as in the quantum case, thereby proving that
the embedding problem too is NP-hard.

\bibliographystylesupporting{apsrev4-1}
\bibliographysupporting{Markov_complexity}

\onecolumngrid  
\clearpage
\twocolumngrid

\end{document}